\documentclass[10pt]{article}
\usepackage{amsmath,amssymb,textcomp}
\usepackage{graphicx,mathrsfs}
\usepackage{multirow}
\usepackage{xcolor}

\newtheorem{theorem}{Theorem}

\newtheorem{corollary}[theorem]{Corollary}

\newtheorem{lemma}[theorem]{Lemma}

\newtheorem{problem}[theorem]{Problem}

\newenvironment{proof}[1][Proof]{\textbf{#1.} }{\ \rule{0.5em}{0.5em}}

\begin{document}

\title{Data-Driven System Identification of Linear Quantum Systems Coupled to Time-Varying Coherent Inputs}

\author{Hendra I. Nurdin\thanks{H. I. Nurdin is with the School of Electrical Engineering and 
Telecommunications,  UNSW Australia,  Sydney NSW 2052, Australia (\texttt{email: h.nurdin@unsw.edu.au})} \and Nina H. Amini \thanks{N. H. Amini is  a Chargée de Recherche CNRS at Laboratoire des Signaux et Systèmes (L2S), CentraleSupélec, 91190 Gif-sur-Yvette, France (\texttt{nina.amini@l2s.centralesupelec.fr})} \and Jiayin Chen\thanks{J. Chen is with the School of Electrical Engineering and 
Telecommunications,  UNSW Australia,  Sydney NSW 2052, Australia (\texttt{email: jiayin.chen@student.unsw.edu.au} )}
}

\date{}

\maketitle

\begin{abstract}
In this paper, we develop a system identification algorithm to identify a model for unknown linear quantum systems driven by time-varying coherent states, based on empirical single-shot continuous homodyne measurement data of the system's output. The proposed algorithm identifies a model that satisfies the physical realizability conditions for linear quantum systems, challenging constraints not encountered in classical (non-quantum) linear system identification. Numerical examples on a multiple-input multiple-output optical cavity model are presented to illustrate an application of the identification algorithm. 
\end{abstract}

\section{Introduction}
Black-box modelling is a modelling paradigm based on learning about a system by observing its response to given inputs, without any prior knowledge of the system's internal structure. It is an important  paradigm in science and engineering, in particular in systems and control. For dynamical systems, black-box modelling is achieved through system identification and has a long rich history \cite{Ljung99}. In system identification, {\em single-shot (stochastic) measurement data} (i.e., a single stochastic observation record) collected from a system of interest is recorded against known inputs injected into it and  a mathematical model, chosen from a class of models with some unspecified parameters, is fitted based on the data. Stochasticity arises due to internal noise in the system as well as measurement noise.

In the quantum context, parameter estimation and versions of black-box modelling of dynamical quantum systems have been considered in various contexts; see, e.g., \cite{Mab96,BY12,ZS14,ZS15,SC17,WDZPY18} and the references therein. Parameter estimation for the class of quantum stochastic input-output models \cite{GZ04,GJ09,GJ08,CKS17}, ubiquitous in various physical platforms such as  quantum optics, quantum electrodynamical (QED) systems and superconducting circuits,  was initiated by Mabuchi \cite{Mab96}. However, the existing methods share one or more of the following features:  (i) they were developed for models other than quantum stochastic input-output models (e.g., closed systems with an unknown Hamiltonian) \cite{BY12,ZS14,ZS15,SC17,WDZPY18}, (ii) use repeated projective measurements and averaging rather than a single continuous measurement record \cite{BY12,ZS14,ZS15,SC17,WDZPY18} or (iii) assume everything is known about the system except for one or a number of unknown parameters \cite{Mab96,ZS15,SC17}. 

Recent works have investigated fundamental aspects of system identification for quantum input-output systems \cite{GY16,GK17,LG17,LGN18} but no empirical methods have yet been developed for system identification using  single-shot continuous measurement data.  Such methods are crucial for practical applications of system identification for quantum input-output systems. This paper will begin to close this gap by initiating the study of empirical system identification for the class of linear quantum systems \cite[\S 6.6]{WM10} \cite{NY17} based on single-shot continuous measurement data, in the spirit of the classical setting \cite{Ljung99}. The possibility  of using single-shot measurement data means that quantum input-output systems, such as linear quantum systems, could potentially be identified much more efficiently compared to other classes of quantum models in term of data collection. 

\noindent \textbf{Notation.} Throughout the paper, we will use the following notation. $X^{\top}$ denotes the transpose of a matrix $X$, $X^{\dag}$ denotes the adjoint of a Hilbert space operator $X$ and if $X=[X_{jk}]$ is a matrix of operators then $X^{\dag}$ is the conjugate transpose of $X$, $X^{\dag} =[X_{kj}^{\dag}]$. $I_n$ will denote an $n \times n$ identity matrix. 

\section{Linear quantum stochastic systems}
Linear quantum stochastic systems, or simply linear quantum systems, are the quantum analogue of linear stochastic systems and represent a collection of quantum harmonic oscillators coupled to one another through a quadratic Hamiltonian as well as being linearly coupled to external bosonic fields. They represent various quantum devices that have linear quantum stochastic evolution in the Heisenberg picture. This includes, for example, optical and superconducting cavities and parametric amplifiers, and gravitational wave interferometers \cite{GZ04,GJ09,CKS17}. They are of interest for linear quantum information processing with quantum Gaussian states and gravitational wave interferometry. 

\begin{figure}[!h]
\centering
\includegraphics[scale=0.43]{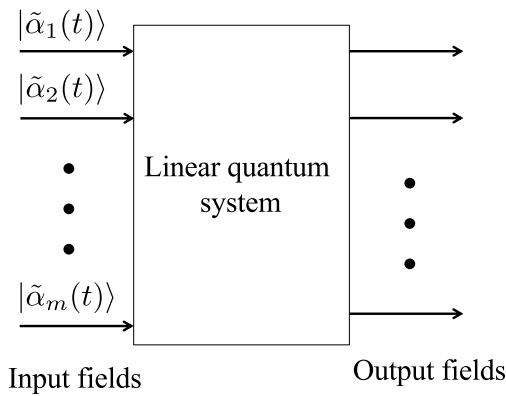}
\caption{ A linear quantum system driven by $m$ fields each in a coherent state}
\label{fig:lqs}
\end{figure}

Linear quantum systems are described by a vector $x =(q_1,p_1,q_2,p_2,\ldots,q_n,p_n)^{\top}$ where $q_j$ and $p_j$ are the position and momentum operators of oscillator $j$ and $n$ is the number of oscillators, a quadratic Hamiltonian $H = \frac{1}{2}x^{\top}  R x$, where $R = R^{\top} \in \mathbb{R}^{2n \times 2n}$, a linear coupling operator $L =Kx$ to $m$ external fields with $K \in \mathbb{C}^{m \times 2n}$, and a scattering matrix $S \in \mathbb{C}^{m \times m}$. When the system is driven by $m$  fields that are in a coherent state with amplitude vector $\tilde{\alpha}(t) = \tilde{\alpha}_R(t) + i \tilde{\alpha}_I (t) \in \mathbb{C}^{m}$, with $\tilde{\alpha}_R(t),\tilde{\alpha}_I(t) \in \mathbb{R}^m$ (see  Fig.~\ref{fig:lqs}), the  joint evolution of the system and field is given by a unitary propagator $U(t)$ solving the Hudson-Parthasarathy quantum stochastic differential equation (QSDE) \cite{KRP92}:
\begin{equation*}
\begin{split}
dU(t) =& \left(-i\left(H+\frac{1}{2}(L+\tilde{\alpha})^{\dag}(L+\tilde{\alpha})\right) dt \right. 
 \vphantom{\frac{1}{2}}  + d\mathcal{A}(t)^{\dag} (L +\tilde{\alpha}(t)) - (L+\tilde{\alpha}(t))^{\dag}d\mathcal{A}(t) 
\left. \vphantom{\frac{1}{2}} + {\rm Tr}((S-I)d\Lambda(t))\right)U(t), 
\end{split}
\end{equation*}
with initial condition $U(0)=I$. In the above QSDE, $\mathcal{A}(t)=[\begin{array}{cccc} \mathcal{A}_1(t) & \mathcal{A}_2(t) & \ldots & \mathcal{A}_m(t) \end{array}]^{\top}$ is the vector of annihilation operators for the $m$ field and $\Lambda(t) =[\Lambda_{jk}(t)]_{j,k=1,\ldots,m}$ (with $\Lambda_{jk}^{\dag} = \Lambda_{kj}$) satisfying the quantum It\={o} product rule:
\begin{align*}
d\mathcal{A}_j(t) d\mathcal{A}_k^{\dag}(t) &= \delta_{jk} dt,\; d\Lambda_{jk}(t)d\Lambda_{uv}(t) = \delta_{ku} d\Lambda_{jv}(t),\;\\
d\Lambda_{jk}(t)d\mathcal{A}^{\dag}_{l}(t) &= \delta_{kl} d\mathcal{A}^{\dag}_j(t), 
\end{align*}
with all other products between $d\mathcal{A}_j(t)$, $d\mathcal{A}^{\dag}_k(t)$ and $d\Lambda_{uv}(t)$ and their adjoints vanishing.

Let $\eta(t)=(\eta_1^q(t),\eta_1^p(t),\eta_2^q(t),\eta_2^p(t),\ldots,\eta_m^q(t),\eta_m^p(t))^{\top}$ with $\eta_j^q(t)=\mathcal{A}_j(t) + \mathcal{A}_j(t)^{\dag},\eta_j^p(t)=-i\mathcal{A}_j(t) + i\mathcal{A}_j(t)^{\dag}$ the amplitude and phase quadratures of the $j$-th field, respectively. The Heisenberg evolution $x(t) = U(t)^{\dag} x U(t)$ of the vector $x$ of position and momentum operators and the vector of output field $y(t)=U(t)^{\dag}\eta(t) U(t)$ are given by the linear QSDE (in the so-called quadrature form \cite[Chapter 2]{NY17}:
\begin{align}
\label{eq:lqs} \begin{split}
dx(t) &= A x(t) dt + B (\alpha(t) dt+ d\eta(t))\\
dy(t) &= Cx(t) dt + D(\alpha(t) dt + d\eta(t)). 
\end{split}
\end{align}
In the above, $A \in \mathbb{R}^{2n \times 2n}$, $B \in \mathbb{R}^{2n \times 2m}$, $C \in \mathbb{R}^{2m \times 2n}$, $D \in \mathbb{R}^{2m \times 2m}$ and $\alpha(t)$ is $2m \times 1$ vector of real functions representing the phase and amplitude quadratures of coherent amplitudes driving the system, $ \alpha(t) = (\tilde{\alpha}_{R,1}(t),\tilde{\alpha}_{I,1}(t),\ldots,\tilde{\alpha}_{R,m}(t),\tilde{\alpha}_{I,m}(t))^{\top}$, where $\tilde{\alpha}_{s,j}$ is the $j$-th component of $\tilde{\alpha}_{s}$, $s \in \{R,I\}$.  Similarly,  $y(t)=(y_1^q(t),y_1^p(t),y_2^q(t),y_2^p(t),\ldots,y_m^q(t),y_m^p(t))^{\top}$ is the output field vector containing the amplitude and phase quadratures of the output fields, where $y_j^q$ and $y_j^p$ denote the amplitude and phase quadratures of the $j$-th field, respectively. Due to quantum constraints, the matrices $A,B,C,D$ need to satisfy the physical realisability constraints \cite{JNP08,NY17}:
\begin{equation*}
 A \mathbb{J}_n +  \mathbb{J}_n  A^{\top} + B  \mathbb{J}_m B^{\top} =0,\; \mathbb{J}_n C^{\top} + B\mathbb{J}_m D^{\top} =0,
\end{equation*}
where $\mathbb{J}_n = I_n \otimes J$ and $J=\left[\begin{array}{cc} 0 & 1 \\ -1 & 0 \end{array}\right]$. If only steady-state measurement data are available, the parameters will only be identifiable up to a similarity transformation, $(A,B,C,D) \rightarrow (VAV^{-1},VB,CV^{-1},D)$ for some real invertible matrix $V$ \cite{GY16,LG17}. With this transformation,  $\mathbb{J}_n$ is replaced with $Z=V \mathbb{J}_n V^{\top}$ and the physical realizability constraints become:
\begin{equation}
 A Z +  Z  A^{\top} + B  \mathbb{J}_m B^{\top} =0\; \mathrm{(I)},\;  ZC^{\top}+ B\mathbb{J}_m D^{\top} =0 \; \mathrm{(II)}.  \label{eq:pr-similar}
\end{equation}
Note that in the above the matrix $Z$ is skew-symmetric, $Z=-Z^{\top}$ and is required to be invertible. 

Information about the system can be obtained by performing measurements on its output. For instance, two basic measurements are  $y_q(t) = (y_1^q(t),y_2^q(t),\ldots,y_m^q(t))^{\top}$ and $y_p(t) = (y_1^p(t),y_2^p(t), \ldots,y_m^p(t))^{\top}$. These measurements are known as {\em homodyne measurements} \cite[\S 4.4]{WM10}. The vector $y_q(t)$ is a homodyne measurement of the amplitude quadrature of the output, while $y_p(t)$ is a homodyne measurement of the phase quadratures. Note that quantum mechanics does not allow simultaneous measurements of $y_q(t)$ and $y_p(t)$ because the elements of these two vectors do not all commute with one another. Thus, it is only meaningful to measure one of these vectors at any time. It follows that,
\begin{eqnarray*}
dy_q(t) &= C_q x(t) dt + D_q (\alpha(t) dt + d\eta(t)),\\
dy_p(t) &= C_p x(t) dt + D_p (\alpha(t)dt + d\eta(t)),
\end{eqnarray*} 
with $C=[\begin{array}{cc} C_q^{\top} & C_p^{\top} \end{array}]^{\top}$ and $D=[\begin{array}{cc} D_q^{\top} & D_p^{\top}\end{array}]^{\top}$. It is possible to perform {\em heterodyne measurement} of $y_q$ and $y_p$ \cite[\S 4.5]{WM10} which would allow noisy simultaneous measurements of $y_p$ and $y_q$ (but they are not true simultaneous measurements of both quadratures).

When continuous measurement is performed on the quantum system, say by continuously measuring $y_q(t)$, the observed system undergoes a stochastic evolution according to the quantum Kalman filtering equation \cite[\S 4.2]{NY17}:
\begin{align*}
 \begin{split}
d\hat{x}^q(t) &= A \hat{x}^q(t) dt+ B \alpha(t) dt + L_q(t)d\nu_q(t) \\
dy_{qm}(t) &= C_q \hat{x}^q(t) dt + D_q \alpha(t) dt + D_qD_q^{\top}d\nu_q(t).
\end{split}
\end{align*}
Here $y_{qm}(t)$ is the measurement stochastic process (which can be mapped from the operator-valued quantum stochastic process $y_q(t)$ via the Spectral Theorem \cite[Theorem 3.3]{BvHJ07}), $\hat{x}^q$ is the conditional expectation of $x^q$ given the measurement $y_{qm}(t)$ \footnote{$\hat{x}^q$  is also the best mean square estimate of $x^q$ based on $y_{qm}(t)$ \cite{BvHJ07,NY17}} and 
\begin{align*}
\nu_q(t)= \left(D_q D_q^{\top}\right)^{-1}\left(y_{qm}(t) - \int_{0}^t (C_q\hat{x}^q(\tau)+D_q \alpha(\tau)) d\tau \right)
\end{align*}
is the so-called  {\em innovation process} of the quantum Kalman filter. Note that $\nu_q(t)$ is a classical standard Wiener process, $\mathbb{E}[\nu_q(t)\nu_q(t')^{\top}]=\min\{t,t'\}I_n$ that is independent of $\hat{x}^q(s)$ for all $0 \leq s \leq t$.  In the quantum Kalman filter, $L_q$ is the Kalman gain and is given by:
\begin{align*}
L_q (t)= Q_q(t) C_q^{\top} + B D_q^{\top},
\end{align*}
where $Q_q(t)=Q_q(t)^{\top} \geq 0$ satisfies the Riccati differential equation (RDE):
{\small
\begin{align*}
\dot{Q}_q(t) &=AQ_q(t)+ Q_q(t) A^{\top} + BB^{\top}-(Q_q(t)C_q^{\top} + B D_q^{\top}) (D_qD_q^{\top})^{-1}(Q_q(t)C_q^{\top} + BD_q^{\top})^{\top}.
\end{align*} 
}

If the system is asymptotically stable (i.e., the matrix $A$ is Hurwitz), the quantum Kalman filter converges to the steady-state quantum Kalman filter 
\begin{align}
\label{eq:innovation-form} \begin{split}
d\hat{x}^q(t) &= A \hat{x}^q(t) dt+ B \alpha(t) dt + L_qd\nu_q(t) \\
dy_{qm}(t) &= C_q \hat{x}^q(t) dt + D_q \alpha(t) dt + D_qD_q^{\top}d\nu_q(t).
\end{split}
\end{align}
where $L_q$ is the steady-state Kalman gain given by
\begin{align}
L_q &= Q_q C_q^{\top} + B D_q^{\top}, \label{eq:L-ss}
\end{align}
and $Q_q=Q_q^{\top} \geq 0$ satisfies the algebraic Riccati equation (ARE):
\begin{equation}\label{eq:ARE}
AQ_q + Q_qA^{\top} + BB^{\top} \notag  -(Q_qC_q^{\top} + B D_q^{\top}) (D_qD_q^{\top})^{-1}(Q_qC_q^{\top} + BD_q^{\top})^{\top} =0. 
\end{equation}
Although the equations above are given for measurement of $y_{qm}(t)$, analogous equations can be obtained when measurement of $y_{pm}(t)$ is made.

\section{Formulation and numerical solution of identification problem}

\subsection{Problem formulation}

In the system identification problem, we are interested in identifying a model of the form \eqref{eq:lqs} but with system matrices not necessarily of the same dimension, since the true dimensions are not known beforehand,  based on the measurement data $y_{qm}(t)$ or $y_{pm}(t)$. In this paper we do not consider heterodyne measurement of $y_q$ and $y_p$ but the approach can be adapted to that case. Throughout, we will consider the system identification problem under the following assumptions:

\noindent \textbf{Assumptions}
\begin{enumerate}
\item The matrix $A$ is Hurwitz.

\item  The data is collected after the system is at steady-state. 

\item The matrix $D$ is known. Hence $D_q$ and $D_p$ are known.
\end{enumerate}

An application of standard identification algorithms using knowledge of the single-shot continuous measurement record, say, $y_{qm}(t)$, would identify a model in the innovation form \eqref{eq:innovation-form}  with system matrices $(\hat{A},\hat{B},\hat{C}_q,\hat{L}_q)$. However, the identified system matrices from these algorithms will not necessarily satisfy the physical realizability constraints  \eqref{eq:pr-similar} as well as the constraints \eqref{eq:L-ss} and \eqref{eq:ARE}.

Suppose that we have identified system matrices $(\hat{A},\hat{B},\hat{C}_q,\hat{L}_q)$ through some classical identification procedure, such as ARMAX modelling or subspace identification \cite{Ljung99,OD96,Qin06}. The remaining problem is to identify system matrices $(\overline{A},\overline{B},\overline{C}_q,\overline{L}_q)$  that do satisfy all the constraints required of a linear quantum system. The following standard results will be useful in the ensuing discussion, we include the proofs here for the sake of completeness. 

\begin{lemma}
\label{lem:uniqueness-homogeneous} Let $\hat{A}$ be Hurwitz. Then the matrix equation $\hat{A} Z + Z\hat{A}^{\top} = 0$, with $Z$ the same dimension as $\hat{A}$, has the unique solution $Z=0$.
\end{lemma}
\begin{proof}
Let $z_j$ denote the $j$-th column of $Z$ and let $\mathrm{vec}(Z)$ be the vectorization of $Z$ by stacking its columns one on top of the other starting with $z_1$ at the very top. The equation $\hat{A} Z + Z\hat{A}^{\top} = 0$ is equivalent to the equation $(\hat{A} \otimes I + I \otimes \hat{A}) \mathrm{vec}(Z)=0$. If $\lambda_1,\lambda_2,\ldots,\lambda_n$ are eigenvalues of $A$ (including their multiplicities), which all have negative real parts, then the eigenvalues of $ \hat{A} \otimes I + I \otimes \hat{A}$ are $\lambda_i+\lambda_j$ for $i,j=1,2,\ldots,n$. Therefore all eigenvalues of $ \hat{A} \otimes I + I \otimes \hat{A}$ also have negative real parts. It follows that the unique solution of $(\hat{A} \otimes I + I \otimes \hat{A}) \mathrm{vec}(Z)=0$ is $\mathrm{vec}(Z)=0$. Therefore, $Z=0$ is the unique solution of $\hat{A} Z + Z\hat{A}^{\top} = 0$.
\end{proof}

\begin{corollary}
\label{cor:uniqueness-pr-1} Let $\hat{A}$ be Hurwitz. Then the matrix equation  $\hat{A}Z + Z \hat{A}^{\top} + B \mathbb{J}_mB^{\top}=0$ has a unique solution $Z$ and this solution is skew-symmetric.
\end{corollary}
\begin{proof}
Following the proof of Lemma \ref{lem:uniqueness-homogeneous}, $\hat{A}Z + Z \hat{A}^{\top} + B \mathbb{J}_mB^{\top}=0$  is equivalent to the equation $(\hat{A} \otimes I + I \otimes \hat{A}) \mathrm{vec}(Z) = -\mathrm{vec}(B \mathbb{J}_m B^{\top})$. By the same argument as in that proof, when $\hat{A}$ is Hurwitz the equation has a unique solution $Z$, corresponding to $\mathrm{vec}(Z)=-(\hat{A} \otimes I + I \otimes \hat{A})^{-1} \mathrm{vec}(B \mathbb{J}_mB^{\top})$. Furthermore,  we can also inspect that  if $Z$ is a solution then so is  $-Z^{\top}$. Therefore, $Z=-Z^{\top}$ and the unique solution must be skew-symmetric.
\end{proof}


In the approach that will be developed below, we first determine $(\overline{A},\overline{B},\overline{C}_q)$ (with a Hurwitz $\overline{A}$) and then solve for the Kalman gain $\overline{L}_q$. Given estimates $(\hat{A},\hat{B},\hat{C}_q)$, we introduce a loss function $\mathcal{L}$ that is nonnegative function of $\Delta A=\overline{A}-\hat{A}$, $\Delta B=\overline{B}-\hat{B}$ and $\Delta C_q = \overline{C}_q-\hat{C}_q$ with the property that $\mathcal{L}(\Delta A, \Delta B, \Delta C_q)=0 \Rightarrow \Delta A=0, \Delta B=0$ and $\Delta C_q=0$. 
 
We formulate a linear quantum system identification problem as follows.

\begin{problem}
\label{sysid_prob}
$$\mathop{\mathrm{minimize}}_{\overline{A}, \overline{B},\overline{C}_q, Z, P} \mathcal{L}(\Delta A, \Delta B, \Delta C_q)$$
subject to
\begin{align}
\label{eq:opt-cons} \begin{split}
& P > 0, \\
& P-P^{\top}=0,\\
& \overline{A}^\top P + P \overline{A} <0,  \\
& \overline{A} Z +  Z  \overline{A}^\top + \overline{B}  \mathbb{J}_m \overline{B}^{\top} =0,  \\ 
& Z\overline{C}_q^{\top}+ \overline{B}\mathbb{J}_m D_q^{\top} =0,  \\
& Z+Z^{\top}=0, \\
& \det(Z)^2 > 0.
\end{split} 
\end{align}
\end{problem}

For the loss function $\mathcal{L}$, we choose a simple quadratic  function,
\begin{align*}
\lefteqn{\mathcal{L}(\Delta A,\Delta B, \Delta C_q)} \\
&= \frac{1}{2}\left( \|\overline{A} - \hat{A}\|_2^2 + \| \overline{B} - \hat{B}\|_2^2 + \| \overline{C}_q - \hat{C}_q\|_2^2\right),
\end{align*}
where $\|X\|_2 = \sqrt{{\rm tr}(X^{\top}X)}$.

After obtaining a solution $(\overline{A}, \overline{B}, \overline{C}_q)$ to the optimization problem \ref{sysid_prob}, we solve for the corresponding Kalman gain $\overline{L}_q$ for the linear quantum system according to \eqref{eq:L-ss} and \eqref{eq:ARE}.

\subsection{Numerical solution}
\label{sec:numerical}
The system identification problem, Problem  \ref{sysid_prob},  formulated in the previous section can be viewed as a matrix polynomial programming problem. The objective function is a quadratic function of matrix variables and all the variables are matrix-valued. This is a formidable non-convex optimization problem for which there is  no known general solution. Here we borrow  a technique proposed in \cite{NJP09} to introduce matrix lifting variables to transform  the original matrix polynomial  programming problem to a rank constrained LMI problem.  The latter problem can be numerically solved with the LMIRank algorithm \cite{LMIRank,OHM06} (run on the Yalmip toolbox for Matlab \cite{YALMIP}) as originally proposed in \cite{NJP09} (see also \cite[\S 5.2.1]{NY17}).  

In the transformation below we will drop the constraint ${\rm det}(Z)^2>0$ as generically this constraint is expected to be satisfied in the sense that the set where $\det(Z)=0$ forms a ``thin set" in the set of  all skew-symmetric matrices in $\mathbb{R}^{2n \times 2n}$; for a discussion of the notion thinness, see, e.g., \cite{NGP16}. To transform the problem we introduce two positive semidefinite symmetric matrix lifting variables $\mathbf{G}_1 \in \mathbb{R}^{(10n+3m) \times (10n+3m)}$ and $\mathbf{G}_2 \in \mathbb{R}^{(4n+2m) \times (4n+2m)}$. We will require these two matrices to satisfy the rank constraints ${\rm rank}(\mathbf{G}_1) \leq 2n$ and ${\rm rank}(\mathbf{G}_2) \leq 2m$. If these matrices do indeed satisfy the rank constraints then we can factorize them as $\mathbf{G}_j = G_j G_j^{\top}$  and identify the block elements of $G_j$ as follows:\\
{\small
\begin{equation}
\label{eq:lmi_var}
\begin{split}
G_1^{\top} & = \begin{bmatrix} I_{2n} & \overline{A}^\top & \overline{A} & \overline{B} & \overline{C}_q^{\top} & Z^{\top} & P^\top \end{bmatrix}, \\
G_2^{\top} & = \begin{bmatrix} I_{2m} & \overline{B}^{\top} & \mathbb{J}_m^{\top} \overline{B}^{\top} \end{bmatrix}.
\end{split}
\end{equation}
}
Now, let ${\bf G}_j(k, l)$ denote the $(k, l)$-th block matrix in ${\bf G}_j$.  If the $\mathbf{G}_j$ matrices satisfy the specified rank constraints then we have the identification ${\bf G}_j(k, l)=G_j(k) G_j(l)^{\top}$, where $G_j(k)$ denotes the $k$-th block element of $G_j$ according to the block partitioning in \eqref{eq:lmi_var}.  In terms of these block matrices the cost function $L$ can be written as
\begin{equation*}
\begin{split}
 \mathcal{L}(\Delta A, \Delta B, \Delta C_q) & = \frac{1}{2} \left( {\rm Tr}[\mathbf{G}_1(2, 2)+ \hat{A} \hat{A}^\top]  + {\rm Tr}[\mathbf{G}_1(4, 4) + \hat{B} \hat{B}^\top] \right. \\
& \left. \hspace*{3em} + {\rm Tr}[\mathbf{G}_1(5, 5) + \hat{C}_q \hat{C}_q^\top]  - 2 {\rm Tr}[\hat{A}^\top \mathbf{G}_1(2, 1)] \right. \\
& \left. \hspace*{3em} - 2{\rm Tr}[\hat{B}^\top \mathbf{G}_1(1, 4)] - 2{\rm Tr}[\hat{C}_q^\top \mathbf{G}_1(5, 1)] \right)
\end{split}
\end{equation*}
and the constraints \eqref{eq:opt-cons} can be written as
\begin{equation*}
\begin{split}
&\mathbf{G}_1(1, 7) - \mathbf{G}_1(7, 1) =0,\\ 
&\mathbf{G}_1(1, 7) \geq \epsilon I_{2n}, \\
&\mathbf{G}_1(3, 7) + \mathbf{G}_1(7, 3) \leq -\epsilon \mathbf{G}_1(1,7),  \\ 
&-\mathbf{G}_1(2, 6) + \mathbf{G}_1(6, 2) + \mathbf{G}_2(3, 2) = 0, \\
&\mathbf{G}_1(6, 5) + \mathbf{G}_1(1, 4) \mathbb{J}_m D_q^\top = 0, \\
&\mathbf{G}_1(1, 6) + \mathbf{G}_1(6, 1) = 0. 
\end{split}
\end{equation*}
where $\epsilon>0$ (we set $\epsilon=10^{-3}$ throughout) and the last constraint ensures the solution for $Z$ returned by the algorithm is skew-symmetric. The constant $\epsilon$ has been introduced to replace strict inequality constraints with non-strict ones, as required for the numerical software packages that will be used. 
From \eqref{eq:lmi_var}, we obtain the following auxiliary constraints on the block elements of $\mathbf{G}_j(k,l)$: \begin{equation*}
\begin{split}
{\bf G}_1(1, 1) - I_{2n} & = 0, \\
{\bf G}_2(1, 1) - I_{2m} & = 0, \\
{\bf G}_1(1, 3) - {\bf G}_1(2, 1) & = 0, \\
{\bf G}_1(1, 4) - {\bf G}_2(2, 1) & = 0, \\
{\bf G}_2(3, 1) - {\bf G}_2(2, 1) \mathbb{J}_m & = 0, \\
{\bf G}_i & \geq 0, \quad i=1,2
\end{split}
\end{equation*}
and the original rank constraints
\begin{equation*}
{\rm rank}({\bf G}_1) \leq 2n, \quad {\rm rank}({\bf G}_2) \leq 2m.
\end{equation*}
We remark that if the above constraints are satisfied the original variables  of the problem  can be recovered from the corresponding block elements of $\mathbf{G}_j$, according to \eqref{eq:lmi_var}. We then solve for the corresponding Kalman gain $\overline{L}_q$ for the identified linear quantum system according to \eqref{eq:L-ss} and \eqref{eq:ARE}.

To solve this rank-constrained LMI problem, we employ the LMIRank algorithm in \cite{OHM06}. The initial guess for the algorithm is chosen to be $\hat{{\bf G}}_j = \hat{G}_j \hat{G}^{\top}_j$, where $\hat{G}_j^{\top}$ is obtained from $G_j^{\top}$ by replacing the variables $(\overline{A},\overline{B}, \overline{C}_q)$ with $(\hat{A}, \hat{B}, \hat{C}_q)$. We set the initial guess for $P$ as a solution to the LMI $\hat{A}^\top P + P \hat{A} < 0, P > 0$ and the initial guess for $Z$ to be $\mathbb{J}_n$.

The LMIRank algorithm only solves a feasibility problem. To minimize the cost function, we employ a standard bisection strategy by including $\mathcal{L}(\Delta A,\Delta B,\Delta C_q) \leq \gamma$ as an additional constraint in the feasibility problem. Starting with an initial guess, we half $\gamma$ each time the LMIRank algorithm returns a feasible solution. Otherwise, we set $\gamma = 1.2 \gamma$.

\section{Numerical examples}
To test the proposed identification method, we will use simulated data of quadrature measurements at the output of a linear quantum system. This can be done in a standard way by generating a sample of a band-limited approximation of the standard white noise vector $\dot{\nu}_q(t)$ (or $\dot{\nu}_p(t)$ depending on the measurement being considered) satisfying $\mathbb{E}[\dot{\nu}_q(t)\dot{\nu}_q(t')^{\top}]=\delta(t-t') I$, and numerically integrating the SDE for the quantum Kalman filter \eqref{eq:innovation-form} with a small sampling time of $T_s$ to generate $\dot{y}_{qm}$ ($y_{qm}$ is just the integral of $\dot{y}_{qm}$). We use time derivatives because classical linear system identification algorithms implemented in Matlab use the derivative  $\dot{y}_{qm}$ as the input data.

As a numerical example, we consider identifying a multiple-input multiple-output optical cavity with position and momentum operators $q$ and $p$. Here $H=\Delta (q^2+p^2)/2$ and $$L = [\begin{array}{ccc} \sqrt{\kappa_1}(q+ i p)/2 & \sqrt{\kappa_2}(q+ i p)/2 & \sqrt{\kappa_3}(q+ i p)/2 \end{array}]^{\top},$$ with $\Delta=10$, $\kappa_1=5$, $\kappa_2=3$, and $\kappa_3=2$, and $S=I_3$, corresponding to the system matrices,  
\begin{equation*}
\begin{split}
A & = \begin{bmatrix} -5 & 20 \\ -20 & -5 \end{bmatrix}, \; C = \begin{bmatrix} 2.2361 & 0 \\ 0 & 2.2361 \\ 1.7321 & 0 \\ 0 & 1.7321 \\ 1.4142 & 0 \\ 0 & 1.4142 \end{bmatrix},\; D  = I_{6}.
\\
B & = \left[\begin{array}{cccccc} -2.2361 & 0 & -1.7321 & 0 & -1.4142 & 0 \\  0 & -2.2361 & 0 & -1.7321 & 0 & -1.4142 \end{array} \right], 
\end{split} 
\end{equation*}

Using a sampling time of $T_s = 10^{-2}$ s, we generate the measurement data from the system ($\dot{y}_{pm}$ and $\dot{y}_{qm}$) for a total time duration of $80$ s, with initial state $\hat{x}^j(0) = 0$, where $j \in \{q,p\}$. The first $20$ seconds of the data is for driving the system to its steady state and is not used for identification. The next $30$ seconds of the data is used for model estimation and the last $30$ seconds is for model validation. The system is excited by a pseudo-random binary sequence (PRBS) generated using the ``idinput'' Matlab command, a persistently exciting input signal \cite[Chapter 13]{Ljung99}. The amplitudes of the PRBS are set to be $\Omega = \{10/\sqrt{T_s}, 50/\sqrt{T_s}, 100/\sqrt{T_s}\}$ to investigate the effect of different signal-to-noise ratio (SNR) in the presence of white noise on the estimated models. We employ subspace identification \cite{OD96,Qin06} through the ``n4sid'' Matlab command to estimate the system matrices. As the order of estimated models is unknown a priori, classical (non-physically realizable) models of state-space dimension $2n \in \{2, 4, 6\}$ are identified and compared using their ``relative energy" contributions, as computed and plotted by the n4sid command. States with small relative energies contribute little to the model accuracy and can be discarded with little impact. Table~\ref{table:q} and Table~\ref{table:p} show the relative energy contributions of estimated classical models using measurement data $\dot{y}_{qm}$ and $\dot{y}_{pm}$, respectively. For all values of $\Omega$, relative energy suggests that the simplest model with $n=1$ is sufficient. As $\Omega$ increases, relative energy for $n=1$ further increases.

From the classical models produced by the subspace identification, we then identify system matrices $\left(\overline{A}_j, \overline{B}_j, \overline{C}_j \right)$ that satisfy all constraints \eqref{eq:opt-cons} of a linear quantum system using the LMIRank algorithm. We observe that the magnitudes of $\hat{B}_j$ estimated by subspace identification are small while the magnitudes of $\hat{C}_j$ are large. To avoid poor numerical conditioning for the LMIRank algorithm, we perform a similarity transformation with $T = 6 \Omega I_{2n}$. This transformation leaves $\hat{A}_j$ unchanged but scales $\hat{B}_{j}$ by $6 \Omega$ and $\hat{C}_{j}$ by $\frac{1}{6 \Omega}$. Using the bisection strategy, LMIRank returns the cost function values $\gamma_j$ tabulated in Table~\ref{table:q} and Table~\ref{table:p}. We compute the Akaike final prediction-error (FPE) as in \cite[Chapter~16]{Ljung99} for the estimated (physically realizable) quantum models obtained after applying the optimization algorithm in Section~\ref{sec:numerical}. The FPE is defined by
\begin{equation*}
\begin{split}
{\rm FPE}_j = \det \left( \frac{1}{N} \sum_{k} e_j(kT_s) e_j(kT_s)^{\top} \right) \frac{1 + d/N}{1 - d/N},  
\end{split}
\end{equation*}
where the summation is over $e_j(k T_s)$ for the validation data (the last 30 s), $N$ is the number of validation data and $d$ is the number of estimated parameters. The prediction error $e_j(kT_s)$ is obtained using the ``resid'' Matlab command. We also report the percentage fit for each output, defined by 
$${\rm Fit}_{j, l} = \left(1- \frac{\sqrt{\sum_{k} e^2_{j, l}(kT_s)}}{\sqrt{\sum_{k} (\dot{y}_{jm,l}(kT_s) - \mu_{jm,l})^2}} \right) \times 100\%,$$
where $l=1,\ldots, m$, $e_{j, l}(kT_s)$ and $\dot{y}_{jm, l}(kT_s)$ are the $l$-th component of $e_j(kT_s)$ and $\dot{y}_{jm}(kT_s)$, and $\mu_{jm, l} = \frac{1}{N} \sum_k \dot{y}_{jm,l}(kT_s)$. The percentage fits are computed using the ``compare'' Matlab command. 

See Table~\ref{table:q} and Table~\ref{table:p} for ${\rm FPE}_j$ and ${\rm Fit}_{j,l}$ for $j\in\{p, q\}$, respectively.

\begin{table}[ht!]
\centering
\caption{Relative energy contributions of estimated classical models, $\gamma_q$, ${\rm FPE}_q$ and ${\rm Fit}_{q, l}$ for estimated quantum models according to measurement data $\dot{y}_{qm}$.}
\label{table:q}
\begin{tabular}{ c | c | c | c | c | c | c | c}
$\Omega$                             & $n$ & Relative  & $\gamma_q$ & ${\rm FPE}_q$  & ${\rm Fit}_{q, 1}$ & ${\rm Fit}_{q, 2}$ & ${\rm Fit}_{q, 3}$ \\
& &  energy &  & $(\times 10^{6})$ & (\%) & (\%) & (\%) \\
\hline\hline
\multirow{3}{1.5em}{$\frac{10}{\sqrt{T_s}}$}   & 1   & $7.61$  & $0.0094$   & $1.11$   & $59.8 $ & $50.1 $ & $42.2 $ \\
                                               & 2   & $5.08$  & $0.65$     & $1.28$   & $59.2 $ & $47.8 $ & $41.6 $ \\
                                               & 3   & $5.00$   & $1.56$     & $1.16$   & $59.8 $ & $50.0 $ & $42.2 $ \\
\hline
\multirow{3}{1.5em}{$\frac{50}{\sqrt{T_s}}$}   & 1   & $9.21$   & $0.004$    & $1.15$   & $91.0 $ & $88.4 $ & $86.3 $ \\
                                               & 2   & $5.08$   & $0.65$     & $6.24$   & $87.5 $ & $86.1 $ & $80.3$ \\
                                               & 3   & $5.00$   & $1.56$     & $1.43$   & $90.6$ & $88.4$ & $85.7$ \\
\hline
\multirow{3}{1.5em}{$\frac{100}{\sqrt{T_s}}$}   & 1   & $9.91$   & $0.001$   & $1.14$   & $95.5$ & $94.2$ & $93.1$ \\
                                                & 2   & $5.08$   & $0.54$    & $1.17$   & $95.5$ & $94.2$ & $93.1$ \\
                                                & 3   & $5.01$   & $1.56$    & $2.46$   & $94.6$ & $93.9$ & $92.0$ \\
\end{tabular}
\end{table}

\begin{table}[ht!]
\centering
\caption{Relative energy contributions of estimated classical models, $\gamma_p$, ${\rm FPE}_p$ and ${\rm Fit}_{p, l}$ for estimated quantum models according to measurement data $\dot{y}_{pm}$.}
\label{table:p}
\begin{tabular}{ c | c | c | c | c | c | c | c}
$\Omega$                             & $n$ & Relative   & $\gamma_p$ & ${\rm FPE}_p$  & ${\rm Fit}_{p, 1}$ & ${\rm Fit}_{p, 2}$ & ${\rm Fit}_{p, 3}$ \\
 & & energy & & $(\times 10^{6})$  & (\%) & (\%) & (\%) \\
\hline\hline
\multirow{3}{1.5em}{$\frac{10}{\sqrt{T_s}}$} & 1   & $7.59$   & $0.015$    & $1.11$   & $58.8 $ & $49.3 $ & $42.9 $ \\
                                             & 2   & $5.12$    & $0.65$     & $1.12$   & $58.0 $ & $48.6 $ & $42.7 $ \\
                                             & 3   & $5.05$    & $3.73$     & $1.13$   & $58.1 $ & $48.8 $ & $41.7 $ \\
\hline
\multirow{3}{1.5em}{$\frac{50}{\sqrt{T_s}}$} & 1   & $9.20$    & $0.004$    & $1.11$  & $91.0 $ & $88.4 $  & $86.1 $ \\
                                             & 2   & $5.12$    & $0.65$     & $3.0$   & $89.2 $ & $86.0 $  & $83.8 $ \\
                                             & 3   & $5.05$    & $3.73$     & $9.5$   & $87.5 $ & $82.5 $  & $80.3 $ \\
\hline
\multirow{3}{1.5em}{$\frac{100}{\sqrt{T_s}}$} & 1   & $9.90$   & $0.001$   & $1.11$   & $95.4 $ & $94.2 $ & $93.0 $ \\
                                              & 2   & $5.12$   & $0.54$    & $1.46$   & $95.3 $ & $94.0 $ & $92.7 $ \\
                                              & 3   & $5.05$   & $3.73$    & $22.4$   & $91.2 $ & $90.8 $ & $88.1$ \\
\end{tabular}
\end{table}
For all values of $\Omega$, estimated physically realizable quantum models with $n=1$ achieve the smallest $\gamma_j$ and ${\rm FPE}_j$, as well as the best percentage fits. As the signal amplitude increases, for $n=1$, $\gamma_j$ decreases from around $0.01$ to $0.001$ and the percentage fits increase from around 50\% to over 90\%. 
In fact, when $\Omega=100/\sqrt{T_s}$, the estimated classical models with system matrices below almost satisfy the physical realizability constraints:
\begin{align*}
\hat{A}_q & = \begin{bmatrix}  -5.22 & -20.05 \\ 19.97 & -4.78 \end{bmatrix},\; \hat{C}_q = \begin{bmatrix}   1.20  & 0.20 \\
                                0.93  & 0.15 \\
                                0.76  & 0.12 \end{bmatrix} \times 10^{4}, \\
\hat{B}_q & = \left[\begin{array}{cccccc} -4.06 & -0.67 & -3.12 & -0.52 & -2.55 & -0.42 \\
                                        -0.61 &  4.04 & -0.49 & 3.12 & -0.37 & 2.59 
 \end{array} \right] \times 10^{-4},
\end{align*}
and
\begin{align*}
\hat{A}_p & = \begin{bmatrix} -4.78  & -20.16 \\
                              19.87  & -5.24 \end{bmatrix},\; \hat{C}_p = \begin{bmatrix} -1.19 & -0.20 \\
                              -0.92 & -0.16 \\
                              -0.76 & -0.13 \end{bmatrix} \times 10^{4}, \\
\hat{B}_p & = \left[\begin{array}{cccccc} -0.67 & 4.06 & -0.53 & 3.13 & -0.44 & 2.54\\
                                        4.01  & 0.75 & 3.08  & 0.54 & 2.59 & 0.45
\end{array}\right] \times 10^{-4}.
\end{align*}
This suggests that when $\alpha(t)$ has sufficiently large amplitude (corresponding to a large SNR ratio of the input signal to the quantum noise) the classical subspace identification algorithm is able to produce identified classical models that are close to being physically realizable linear quantum models. 
To obtain the physically realizable system matrices, we decompose $Z_j$ as $Z_j = V_j \mathbb{J}_n V^{\top}_j$ for $j\in \{q, p\}$, where
\begin{equation*}
\begin{split}
Z_q & = \begin{bmatrix}  0 & -1.193 \\ 1.193 & 0 \end{bmatrix}, Z_p = \begin{bmatrix} 0  &  -1.189 \\ 1.189 & 0 \end{bmatrix}, \\
V_q & = \begin{bmatrix} 1.09 & 0 \\ 0 & -1.09 \end{bmatrix}, V_p = \begin{bmatrix} 0 & -1.09 \\ -1.09 & 0 \end{bmatrix}.
\end{split}
\end{equation*}
Then the corresponding physically realizable system matrices are $(\overline{A}_j, \overline{B}_j, \overline{C}_j) = (V_j^{-1} \overline{A}'_j V_j, V^{-1}_j \overline{B}'_j, \overline{C}'_j V_j)$, where $\overline{A}'_j, \overline{B}'_j, \overline{C}'_j$ are solutions returned by the LMIRank algorithm. Based on measurement data $\dot{y}_{qm}$, we obtain
\begin{equation*}
\begin{split}
\overline{A}_q & = \begin{bmatrix}  -5.21 &  20.05 \\ -19.97 & -4.77 \end{bmatrix}, \overline{C}_q = \begin{bmatrix} 2.20  & -0.36 \\
                                   1.70  & -0.28 \\
                                   1.40  & -0.23 \end{bmatrix}, \\
\overline{B}_q & =\left[\begin{array}{cccccc}  -2.22  & -0.36   & -1.71  & -0.28  & -1.40  & -0.23 \\
                                               0.34   & -2.20   & 0.27   & -1.70  & 0.20   & -1.40 \end{array} \right], \\
\overline{L}_q &= \begin{bmatrix} -2.04  & -0.75  & 0.68 \\
                                    -0.34  & 0.34 & -0.89 \end{bmatrix} \times 10^{-2}.
\end{split}
\end{equation*}
Based on measurement data $\dot{y}_{pm}$, we obtain
\begin{equation*}
\begin{split}
\overline{A}_p & = \begin{bmatrix} -5.23 & 19.88 \\ -20.17 & -4.77 \end{bmatrix}, \\
\overline{B}_p & = \left[\begin{array}{cccccc}  -2.19 & -0.40 & -1.68 & -0.29 & -1.38 & -0.25\\
                                                 0.37 & -2.23 & 0.28  & -1.71  & 0.24 & -1.42 
\end{array}\right],\\
\overline{C}_p & = \begin{bmatrix}   0.37  & 2.18 \\
                                     0.29  & 1.68 \\
                                     0.24  & 1.38 \end{bmatrix}, \\
                                     \overline{L}_p & = \begin{bmatrix} -0.33 & 1.43 & 0.81 \\  -2.06 & -1.45 & -2.09 \end{bmatrix} \times 10^{-2}.
\end{split}
\end{equation*}
Furthermore, using the ``resid'' Matlab command, we observe that the residuals $e_j(k T_s)$ of the estimated quantum models are independent of the inputs and the residuals show no autocorrelation (within 99\% confidence interval); see Fig.~\ref{fig:sample_acf} for the residual sample autocorrelation and \cite{Ljung99} for further discussions on residual diagnostics. Fig.~\ref{fig:prediction} plots the predicted outputs of the quantum model with $\Omega=100/\sqrt{T_s}$ and $n=1$ for the first 100 validation data.

\begin{figure}[!ht]
\centering
\includegraphics[scale=0.98]{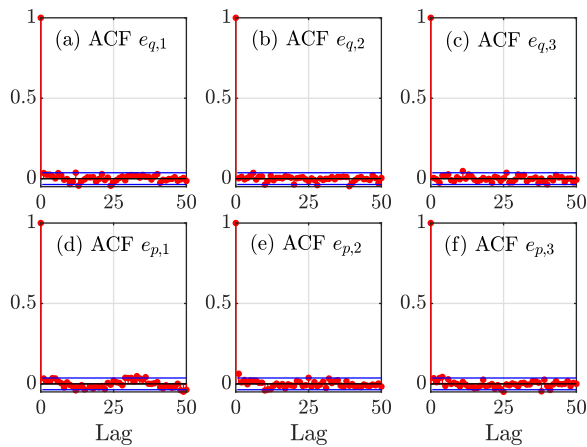}
\caption{Residual sample autocorrelation of the estimated quantum models for (a) $e_{q, 1}(kT_s)$, (b) $e_{q, 2}(kT_s)$, (c) $e_{q, 3}(kT_s)$, (d) $e_{p,1}(kT_s)$, (e) $e_{p, 2}(kT_s)$ and (f) $e_{p, 3}(kT_s)$. Horizontal blue lines are the 99\% confidence bounds. We show the sample autocorrelation up to 50 lags for illustrative purposes and we observe no sample autocorrelation (within 99\% confidence interval) for higher lags.}
\label{fig:sample_acf}
\end{figure}

\begin{figure}[!ht]
\centering
\includegraphics[scale=0.92]{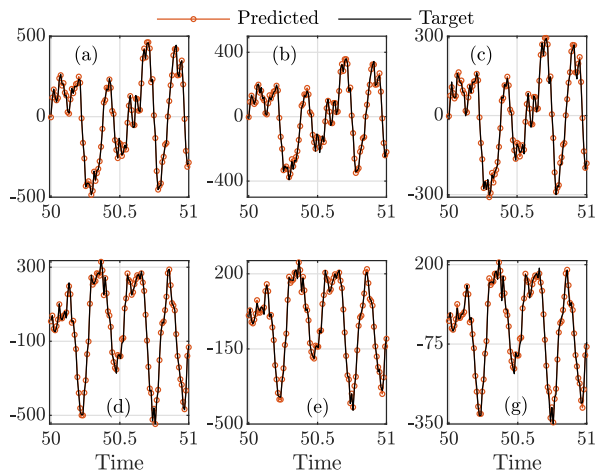}
\caption{Predicted outputs of the estimated quantum model with $\Omega=100/\sqrt{T_s}$ and $n=1$ against target outputs on the first 100 validation data points, for (a) $\hat{\dot{y}}_{qm,1}$ against $\dot{y}_{qm,1}$, (b) $\hat{\dot{y}}_{qm,2}$ against $\dot{y}_{qm,2}$, (c) $\hat{\dot{y}}_{qm,3}$ against $\dot{y}_{qm,3}$, (d) $\hat{\dot{y}}_{pm,1}$ against $\dot{y}_{pm,1}$, (e) $\hat{\dot{y}}_{pm,2}$ against $\dot{y}_{pm,2}$ and (f) $\hat{\dot{y}}_{pm,3}$ against $\dot{y}_{pm,3}$. Here $\hat{\dot{y}}_{jm, l}$ is the $l$-th component of predicted output $\hat{\dot{y}}_{jm}$ for $j \in \{q, p\}$.}
\label{fig:prediction}
\end{figure}

We remark that the values of $\gamma_j$, ${\rm FPE}_j$ and ${\rm Fit}_{j,l}$ differ for different measurement data $j\in\{q, p\}$. This is due to subspace identification returning different identified system matrices for $(\hat{A}_p, \hat{B}_p)$ and $(\hat{A}_q, \hat{B}_q)$. The two estimates are not expected to be the same as they are estimated using distinct measurement data that are in turn also generated, in general, through distinct stochastic evolutions. It may be possible to develop a technique to merge these two models together to obtain a single identified model but this is beyond the scope of the present work and is a theme for future research.

\addtolength{\textheight}{-0.1cm}

\section{Conclusion}   
In this paper, based on appropriate assumptions on the system to be identified, we develop a method to identify linear quantum system models based on single-shot continuous stochastic homodyne measurement data generated by the output of unknown linear quantum systems driven by known coherent input fields. The approach involves a two-step procedure. First  a (non-physically realizable) classical linear stochastic model is identified using well-established classical system identification algorithms. Then a polynomial matrix feasibility  problem is solved to obtain a physically realizable linear quantum system model that is in a sense close to the identified classical stochastic model. We develop a numerical algorithm for solving the polynomial matrix feasibility  problem by   adopting a matrix lifting technique previously used to numerically solve the coherent quantum LQG problem \cite{NJP09}.

We demonstrate our approach in a numerical example. The numerical algorithm is able to identify a multiple-input multiple-output optical cavity based on simulated single-shot homodyne measurement data for varying amplitudes of the coherent input vector $\alpha(t)$. Although classical identification algorithms cannot in general generate physically realizable linear quantum models, our numerical examples indicate that for $\alpha(t)$ with sufficiently high amplitude the classical identified models produced by classical subspace identification can be close to being physically realizable. That is, the identified system matrices almost satisfy the physical realizability constraints of  linear quantum systems. However, in practice,  high amplitude inputs may not be achievable or consume too much energy  to be generated. The case of much practical interest is the one with lower power inputs and this is where the method developed here will be of interest.

The current work assumes the simplification of knowing the output feedthrough matrix $D$, which in general is not the case. Future work can include generalizing the proposed approach to remove this assumption, developing improved numerical algorithms and  finding a method to combine the two models obtained by different measurement quadratures in order to identify a single model. 
 
\bibliographystyle{ieeetran}
\bibliography{qsysid}

\end{document}